\documentclass[aps,superscriptaddress,twocolumn,twoside,floatfix,prl,a4paper]{revtex4-2} 

\usepackage{booktabs}
\usepackage{multirow}
\usepackage{times}
\usepackage{epsfig}
\usepackage{amsfonts}
\usepackage{amsmath}
\usepackage{amssymb,amsthm}
\usepackage{xcolor,colortbl}
\usepackage{multirow}
\usepackage{braket}
\usepackage{latexsym}
\usepackage{amsfonts}
\usepackage{mathrsfs}
\usepackage{natbib}
\usepackage{verbatim}
\usepackage{gensymb}
\usepackage{caption}
\usepackage{caption}
\usepackage{subcaption}
\usepackage{ragged2e}
\DeclareCaptionJustification{justified}{\justifying}
\usepackage{blkarray}
\usepackage{graphicx}
\newtheorem{theorem}{Theorem}

\newtheorem{definition}{Definition}
\newtheorem{proposition}{Proposition}

\usepackage[colorlinks=true,linkcolor=blue,citecolor=magenta,urlcolor=blue]{hyperref}
\allowdisplaybreaks

\newcommand{\tr}{\operatorname{Tr}}

\begin{document}

\title{Principle of information causality rationalizes quantum composition}

\author{Ram Krishna Patra}
\affiliation{Department of Physics of Complex Systems, S. N. Bose National Center for Basic Sciences, Block JD, Sector III, Salt Lake, Kolkata 700106, India.}

\author{Sahil Gopalkrishna Naik}
\affiliation{Department of Physics of Complex Systems, S. N. Bose National Center for Basic Sciences, Block JD, Sector III, Salt Lake, Kolkata 700106, India.}

\author{Edwin Peter Lobo}
\affiliation{Laboratoire d’Information Quantique, Université libre de Bruxelles (ULB), Av. F. D. Roosevelt 50, 1050 Bruxelles, Belgium}

\author{Samrat Sen}
\affiliation{Department of Physics of Complex Systems, S. N. Bose National Center for Basic Sciences, Block JD, Sector III, Salt Lake, Kolkata 700106, India.}

\author{Govind Lal Sidhardh}
\affiliation{Department of Physics of Complex Systems, S. N. Bose National Center for Basic Sciences, Block JD, Sector III, Salt Lake, Kolkata 700106, India.}

\author{Mir Alimuddin}
\affiliation{Department of Physics of Complex Systems, S. N. Bose National Center for Basic Sciences, Block JD, Sector III, Salt Lake, Kolkata 700106, India.}

\author{Manik Banik}
\affiliation{Department of Physics of Complex Systems, S. N. Bose National Center for Basic Sciences, Block JD, Sector III, Salt Lake, Kolkata 700106, India.}

\begin{abstract}
Principle of information causality, proposed as a generalization of no signaling principle, has efficiently been applied to outcast beyond quantum correlations as unphysical. In this letter we show that this principle when utilized properly can provide physical rationale towards structural derivation of multipartite quantum systems. In accordance with no signaling condition state and effect spaces of a composite system can allow different possible mathematical descriptions even when description for the individual systems are assumed to be quantum. While in one extreme, namely the maximal tensor product composition, the state space becomes quite exotic and permits composite states that are not allowed in quantum theory, the other extreme -- minimal tensor product composition -- contains only separable states and the resulting theory allows only Bell local correlation. As we show, none of these compositions does commensurate with information causality, and hence get invalidated to be the bona-fide description of nature. Information causality, therefore, promises information theoretic derivation of self-duality of state and effect cones for composite quantum systems.  
\end{abstract}

\maketitle
{\it Introduction.--} Quantum mechanics is the most effective theory to describe almost all the natural phenomena. Its successful marriage with information theory leads to numerous novel technologies and promises a lot more \cite{Deutsch2020}. However, the theory, starting from its inception, prompts huge debates regarding its interpretation \cite{Interpretation1,Interpretation2,Interpretation3,Interpretation4,Interpretation5,Interpretation6,Interpretation7,Interpretation8} that persists till date \cite{Tegmark1998,Schlosshauer2013}. Quantum formalism starts with abstract mathematical description of Hilbert space and cries for its physical justification. The celebrated no-signaling (NS) principle, that prohibits instantaneous communication between distant parties, cannot serve the purpose alone. It allows a broad variety of mathematical models as the possible candidate of the theory of nature. Interestingly, inspired by the studies in quantum information theory, during the recent past, several novel principles have been proposed to circumvent the limitation of NS principle \cite{Brassard2006,vanDam2013,Pawlowski2009,Navascues2010,Fritz2013,Amaral2014,Dallarno2017}. These new principles quite efficiently identify some beyond quantum NS correlations as unphysical and thus adduce physical justification(s) to quantum correlations.  

In this letter we analyze one of the intriguing principles called information causality (IC), proposed nearly a decade back \cite{Pawlowski2009}. IC can be envisaged as a generalization of the NS condition. It limits the information gain that a receiver (say Bob) can reach about a previously unknown to him data set of a sender (say Alice), by using two types of resources: (i) all his local resources that might be correlated with the sender  which in this work we will call Type-I resource, and (ii) some physical system carrying bounded amount of information from Alice to Bob, Type-II resource. Both these resources can further be of different kinds -- classical, quantum, and beyond quantum; and IC principle provides a way to test their physicality. Although the correlated resources by themselves have no communication utility, as shown in the seminal superdense coding paper \cite{Bennett1992}, a quantum correlation {\it viz.} entanglement can double up the communication capacity of a quantum channel. The power of entanglement, however, is limited in a way as it cannot enhance communication capacity of a classical channel. Principle of IC generalizes this no-go by limiting Bob's information gain to be at most $m$ bits when $m$ classical bits are communicated by Alice to him and he is allowed to use any of his local resources that might be correlated with Alice. Quite interestingly several NS correlations violate this principle and thus considered as unphysical  \cite{Ahanj2010,Das2013,Gazi2013,Miklin2021}. In essence, restricting the Type-II resources to be classical, the IC principle discards some of the Type-I  resources as unphysical. 

Here we study the reverse scenario, {\it i.e.}, Type-II resource is varied among several possibilities while restricting the Type-I resource to be classical. In particular, we ask the question whether unphysicality of some Type-II resources can be established through information principles. Interestingly, as we will show, principle of IC becomes effective in this case too. We consider the scenarios where individual systems are assumed to be quantum, but their composition can be modelled by any theory that satisfy the NS condition. Even for two quantum systems several consistent compositions are possible among which quantum theory is one of the examples. The state space of the resulting system lies in between two extremes -- maximal tensor product state space and minimal tensor product state space \cite{Namioka1969}. While the maximal one grants exotic joint states that are not allowed in quantum theory, the minimal one allows only separable states. It turns out that the system obtained through maximal tensor product of two elementary quantum violates the IC principle. This is quite remarkable as all the NS correlations obtained from beyond quantum states are in fact quantum simulable and hence cannot yield beyond quantum nonlocal correlation \cite{Barnum2010}. We then show that minimal tensor product composition is also not compatible with IC principle and hence gets excluded. This is even more striking as the resulting theory allows only separable states and hence in Bell type experiments no nonlocal correlation is possible. Our work establishes a novel yet unexplored proficiency of the IC principle as it discards extreme descriptions of composite quantum systems, and thus promises physical rationale for quantum composition.

{\it Compositions of elementary quantum.--} Within the mathematical framework of generalized probability theory (GPT) an elementary system $\mathcal{S}$ is specified by the tuple of normalized state and effect spaces, {\it i.e.} $\mathcal{S}\equiv(\Omega,\mathcal{E})$ \cite{Barrett2007,Chiribella2011,Barnum2011,Masanes2011,Plavala2021}. Sometimes it is convenient to deal with unnormalized states and effects that form convex cones embedded in some $\mathbb{R}^n$. A GPT also captures the description of the composite system $\mathcal{S}^{AB}\equiv(\Omega^{AB},\mathcal{E}^{AB})$ consisting of component subsystems $\mathcal{S}^{A}\equiv(\Omega^{A},\mathcal{E}^{A})$ and $\mathcal{S}^{B}\equiv(\Omega^{B},\mathcal{E}^{B})$. Under the restriction of NS and local tomography \cite{Barnum2012} the composite state space $\Omega^{AB}$ lies in between two extremes -- (i) the maximal tensor product state space and (ii) the minimal tensor product state space \cite{Namioka1969}. For instance, the state cone of a quantum system associated with a Hilbert space $\mathcal{H}$ is the set of positive semidefinite operators $\mathcal{P}(\mathcal{H})\subset\mathcal{L}(\mathcal{H})$ acting on $\mathcal{H}$, whereas the normalized states are the set of density operators $\mathcal{D}(\mathcal{H})$; here $\mathcal{L}(\mathcal{H})$ denotes the set of all linear operators acting on $\mathcal{H}$. For two quantum systems associated with Hilbert spaces $\mathcal{H}^A$ and $\mathcal{H}^B$ respectively, the state cone for maximal tensor product system is given by
\begin{align}
\nonumber
\Omega^{AB}_+[\text{max}]&:=\{\mathcal{W}\in\mathcal{L}(\mathcal{H}^A\otimes\mathcal{H}^B)|\tr[\mathcal{W}(\pi_A\otimes\pi_B)]\ge0 \\
&~\forall \pi_A \in\mathcal{P}(\mathcal{H}^A)~ \text{and}~~\forall \pi_B\in\mathcal{P}(\mathcal{H}^B)\}.
\end{align}
States in this theory are called positive on pure tensors (POPT) states \cite{Barnum2010}. All the quantum states are POPT, but there are POPT states that are not allowed quantum states. The effect cone is constructed in accordance with the no-restriction hypothesis that allows all mathematically consistent effects in the theory \cite{Chiribella2010}:
\begin{align}
\nonumber
\mathcal{E}^{AB}_+[\text{max}]&:=\{\pi~|~\pi=\sum_i\pi_i^A\otimes\pi_i^B;~ \pi^A_i\in\mathcal{P}(\mathcal{H}^A)~\\
&	\&~\pi^B_i\in\mathcal{P}(\mathcal{H}^B)\}.
\end{align}
As it turns out $\mathcal{E}^{AB}_+[\text{max}]$ also forms a cone which is dual to the state cone $\Omega^{AB}_+[\text{max}]$.
In the other extreme, minimal tensor product contains only separable states, but the effect space gets enlarged here. More particularly, the role of state and effect cones of maximal tensor product are interchanged in the minimal case, {\it i.e.},
\begin{align}
\Omega^{AB}_+[\text{min}]:=\mathcal{E}^{AB}_+[\text{max}]~\& ~\mathcal{E}^{AB}_+[\text{min}]:=\Omega^{AB}_+[\text{max}].
\end{align}
Quantum composition lies in between the maximal and mininal case, and the state and effect cones become self dual, {\it i.e.}, $\Omega^{AB}_+[Q]=\mathcal{P}(\mathcal{H}^A\otimes\mathcal{H}^B)=\mathcal{E}^{AB}_+[Q]$. An instructive image for theses different composite structures is shown in Fig. \ref{fig2}.
\begin{figure}[t!]
	\centering
	\includegraphics[width=0.45\textwidth]{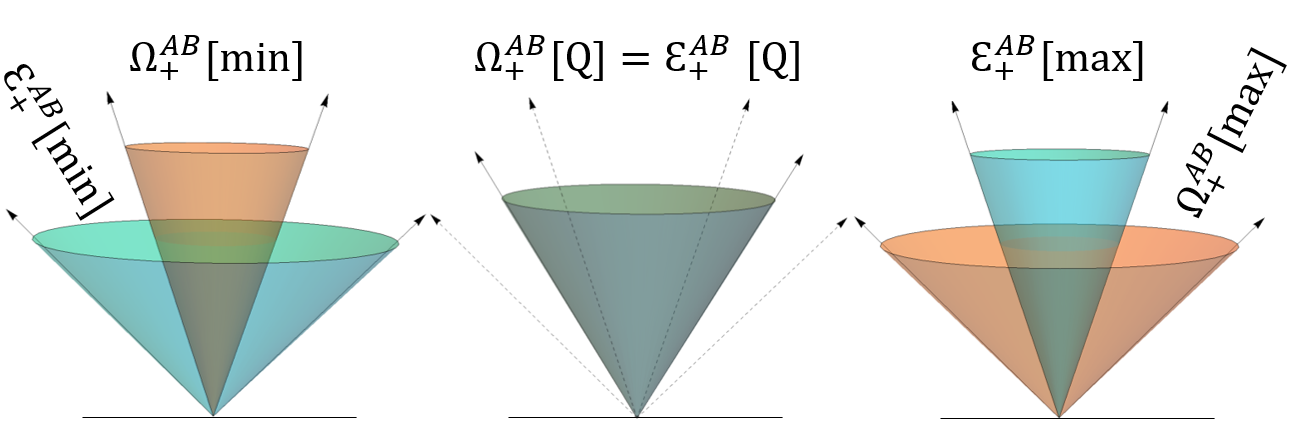}
	\caption{Different possible compositions of two elementary quantum systems. Left: minimal tensor product composition -- allows only separable states but effect cone is bigger than quantum effect cone. Right: maximal tensor product composition -- allows only separable effect but state cone is bigger than quantum. Middle: quantum composition, state and effect cones are identical (self dual).}
	\label{fig2}
\end{figure}

Here we recall the concept of information capacity \cite{Muller2012} (see also \cite{Arai2019,Naik2022,Sen2022}), which is relevant for studying communication utility of a GPT system.
\begin{definition}
[CMP {\bf 316}, 441 (2012)]
Information capacity $\mathcal{I}(\mathcal{S})$ of a system $\mathcal{S}$ is the maximum number of states that can be perfectly distinguished in a single shot measurement.
\end{definition}
In standard information-theoretic unit information capacity can be quantified in terms of bit. If maximum $d$ number of states are perfectly distinguishable by single-shot measurement then $\mathcal{I}(\mathcal{S})=\log d$ bit. Information capacity is closely related to the amount of classical information that can be transferred perfectly by sending one copy of the system. For instance, the state space of a two level classical and quantum systems are the $1$-simplex and unit sphere in $\mathbb{R}^3$ respectively, and both have information capacity $2$ (or $\log2$ bit). However, it might happen that the Holevo capacity $\Theta(\mathcal{S})$ of a system $\mathcal{S}$ exceeds its information capacity. By Holevo capacity of a system we mean the optimal rate of transfer of classical information when asymptotically many copies of the system are in use. While finding Holevo capacity, the optimization is performed over all possible encodings and decodings allowed in the composite theory, {\it i.e.}, $\Theta(\mathcal{S}):=\lim_{k\to\infty}\frac{1}{k}\log\left[\mathcal{I}(\mathcal{S}^{\otimes k})\right]$. It turns out to be same as information capacity whenever information capacity is additive, which follows from the results in \cite{Frenkel2015}. For quantum composition it is known to be additive. On the other hand, as shown in \cite{Massar2014}, for the toy models with polygonal state space \cite{Janotta2011} have information capacity $1$ bit, but for odd gon cases the Holevo capacity is larger than $1$ bit. Our next result shows that Holevo capacity in minimal tensor product composition is additive.
\begin{proposition}\label{prop1}
The information capacity of the minimal composition of $k$ local quantum systems, each of dimension $d$, is $d^k$, {\it i.e.}, $\mathcal{I}
\left((\mathbb{C}^d)^{\otimes_{\min}k}\right) = d^k$.
\end{proposition}
\begin{proof}
Assume that $\mathcal{I}
\left((\mathbb{C}^d)^{\otimes_{\min}k}\right) = N$. Clearly, $N\ge d^k$ since there are $d^k$ product states in $(\mathbb{C}^d)^{\otimes_{\min}k}$ that can be perfectly distinguished by product measurements. Furthermore, since we have assumed $\mathcal{I}
\left((\mathbb{C}^d)^{\otimes_{\min}k}\right) = N$, there exist $N$ pure product states $\{\ket{a_1}, \cdots , \ket{a_N} \}$ and a measurement $\{E_1,\cdots,E_N~|~\sum_{i=1}^{N} E_i =\mathbf{1}_{d}^{\otimes_{k}}\}$ such that $\tr(E_i \ket{a_j}\bra{a_j})=\delta_{ij} ~~\forall i,j$. Let $P_i$ be the projector on the orthogonal support of $\ket{a_i}\bra{a_i}$, {\it i.e.}, $P_i = \mathbf{1}_{d}^{\otimes_{k}} - \ket{a_i}\bra{a_i}$. It is easy to show that $P_i$ is a separable operator. To see this, first introduce an orthonormal basis for $\mathcal{H}_{d}^{\otimes_{k}}$ consisting of pure product states such that $\ket{a_i}\bra{a_i}$ is one of the elements of the basis set. Then, $P_i$ can be written as the sum of projectors on all the basis vectors except $\ket{a_i}\bra{a_i}$. We have,
\begin{align*}
d^k &= \tr(\mathbf{1}_{d}^{\otimes_{k}}) = \sum_{i=1}^{N} \tr(E_i)= \sum_{i=1}^{N} \tr[E_i    \mathbf{1}_{d}^{\otimes_{k}}]\nonumber\\
&= \sum_{i=1}^{N} \tr[E_i   (  \ket{a_i}\bra{a_i} + P_i) ]\nonumber\\
&= \sum_{i=1}^{N} \tr[E_i  \ket{a_i}\bra{a_i} ]  + \sum_{i=1}^{N} \tr[E_i P_i ]\nonumber\\
&\ge \sum_{i=1}^{N} \tr[E_i  \ket{a_i}\bra{a_i}] = \sum_{i=1}^{N} \delta_{ii} = N.
\end{align*}
The inequality follows from the fact that $P_i$ is a separable operator, and thus $\tr[E_i P_i] \ge 0$. We, therefore, conclude that $N= \mathcal{I}
\left((\mathbb{C}^d)^{\otimes_{\min}k}\right) = d^k$.
\end{proof}
As an immediate implication it follows that the Holevo capacity of minimal tensor product of local quantum systems are same as their information capacity, {\it i.e.},
$\Theta(\mathbb{C}^d_{\min}):=\lim_{k\to\infty}\frac{1}{k}\log\left[\mathcal{I}((\mathbb{C}^d)^{\otimes_{\min} k})\right]= \lim_{k\to\infty}\frac{1}{k}\log(d^k) = \log d$. A similar result, as stated in the next proposition, also holds for maximal composition of locally quantum systems.
\begin{proposition}\label{prop2}
The information capacity of the maximal composition of $k$ local quantum systems, each of dimension $d$, is $d^k$, {\it i.e.}, $\mathcal{I}
\left((\mathbb{C}^d)^{\otimes_{\max}k}\right) = d^k$.
\end{proposition}
Proof of this proposition follows a similar reasoning as of Proposition \ref{prop1}. However, for the sake of completeness we detail the proof in Appendix {\bf A}. The objective of the present work is to see whether these extreme compositions can be discarded from some physical ground. In the following we will invoke the IC principle to this aim. Before proceeding to our next results we briefly recall the IC principle.  

{\it Information Causality Game.--}
\begin{figure}[t!]
	\centering
	\includegraphics[width=0.48\textwidth]{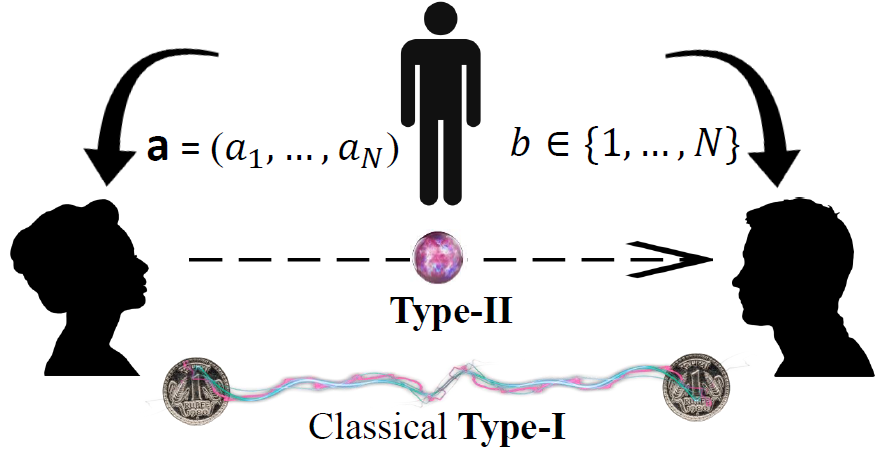}
	\caption{(IC-N game) Referee provides a randomly chosen N-bit string $\mathbf{a}\in\{0,1\}^N$ to Alice and a random value $b\in\{1,\cdots,N\}$ to Bob. Bob's aim is to correctly guess the $b^{th}$ bit of Alice's string, {\it i.e.}, to yield the outcome $\beta=a_b$. Alice and Bob can share classical Type-I resources, {\it i.e.}, classical correlation. Additionally she can send some physical system (Type-II resource) carrying bounded amount of information to Bob.}
	\label{fig1}
\end{figure}
It is instructive to understand the IC principle through the distributed version of random access code task \cite{Ambainis2002}, which we recall from the original IC paper (see Fig.\ref{fig1}). Alice receives a string of $N$ independent bits, $\mathbf{a}=(a_1,\cdots,a_N)$ randomly sampled from $\{0,1\}^N$, and Bob receives a random value of $b\in\{1,\cdots,N\}$. Bob's aim is to correctly guess the value of the $b^{th}$ bit of Alice {\it i.e,} $a_b$. Alice and Bob may possess some classical shared randomness, but unlike the original paper, they do not share any quantum entanglement or beyond quantum correlation ({\it eg.} the Popescu-Rohrlich correlation \cite{Popescu1994}). We call this IC-N game, and the efficiency of the collaborative strategy is quantified through the quantity 
\begin{align}
	I_N=\sum_{k=1}^NI\left(a_k:\beta~|~b=k\right),
\end{align}
where $I\left(a_k:\beta~|~b=k\right)$ is the Shannon's mutual information between the $k^{th}$ bit of Alice and Bob's guess $\beta$, computed under the condition that Bob has received $b=k$. If Alice communicates $m$-cbits to Bob then the efficiency is upper bounded by $I_N\le m$. Furthermore, $I_N\le m$ turns out to be a necessary condition for IC in an arbitrary theory where mutual information satisfies three abstract properties -- (a) consistency, (b) data processing inequality, and (c) chain rule \cite{Pawlowski2009}. As pointed out by the authors in \cite{Pawlowski2009}, IC holds true even if the quantum bits are transmitted provided that they are disentangled from the systems of the receiver -- a consequence of Holevo bound, which limits information gain after transmission of $m$ such qubits to $m$ classical bits. While playing the IC-N game by sending some GPT system from Alice to Bob in presence of classical correlation shared in advance, the precise formulation of the IC, that we will use in this work, reads as $I_N \le \Theta$, where $\Theta$ is the Holevo capacity of the communicated system in consideration. In Appendix {\bf B} we discuss an example of toy theory that leads to violation of IC principle.

{\it Unphysicality of extremal quantum compositions.--} Let us consider the IC-$3$ game. If Alice communicates $2$-cbits or $2$-qubits (in quantum composition) to a classically correlated Bob then the necessary condition of IC is satisfied, {\it i.e.}, efficiency of their collaborative strategy $I_3$ satisfies the bound $I_3\le2$. Our next results show that this is not true if one considers extremal tensor product compositions between two qubits.
\begin{theorem}\label{theo1}
Minimal tensor product composition of two elementary quantum violates IC principle.
\end{theorem}
\begin{proof}
In minimal tensor product theory Alice can use only the product states of $\mathbb{C}^2\otimes\mathbb{C}^2$ system for encoding her strings. The exact encodings is shown in Table \ref{tab1}. Since $\mathcal{E}^{AB}[\text{min}]>\mathcal{E}^{AB}[Q]>\mathcal{E}^{AB}[\text{max}]$, Bob's decodings can comprise of exotic measurements that are not allowed in quantum theory. We consider the decoding measurements $\mathcal{M}_i:=\{E_i,\bar{E}_i:=\mathbf{1}_4-E_i\}$, where
\begin{align*}
E_1:=\begin{pmatrix}1&0&0&\frac{1}{2}\\0&0&\frac{1}{2}&0\\0&\frac{1}{2}&0&0\\\frac{1}{2}&0&0&1\end{pmatrix},~&E_2:=\begin{pmatrix}0&0&0&\frac{1}{2}\\0&1&\frac{1}{2}&0\\0&\frac{1}{2}&1&0\\\frac{1}{2}&0&0&0\end{pmatrix},\\
E_3:=\mathbf{1}_2&\otimes\frac{1}{2}(\mathbf{1}_2+\sigma_z),
\end{align*}
and the bit value is guessed to be $0$ if the effect $E_i$ clicks and $1$ if $ \bar{E}_i$ clicks. Being a valid quantum measurement, $\mathcal{M}_3$ is also an allowed measurement in minimal tensor product theory. Furthermore, on an arbitrary two-qubit product state $\rho_{mn}=\frac{1}{2}(\mathbf{1}_2+\vec{m}.\sigma)\otimes\frac{1}{2}(\mathbf{1}_2+\vec{n}.\sigma)$ we have,  $\tr[E_j\rho_{mn}]=\frac{1}{2}[1+m_xn_x+(-1)^{(j+1)}m_zn_z]\ge0$ and $\tr[\bar{E}_j\rho_{mn}]=\frac{1}{2}[1-m_xn_x-(-1)^{(j+1)}m_zn_z]\ge0$ for $j\in\{1,2\}$, that in turn assure $\mathcal{M}_1~\&~\mathcal{M}_2$ to be bona-fide measurements in minimal tensor product theory. According to Proposition \ref{prop1} the Holevo capacity of this system is 2 bits, and hence the theory violates IC principle since the aforesaid encoding-decoding strategy yields $I_3\approx2.19>2$. This concludes the proof. 
\end{proof}
\begin{center}
	\begin{table}[t!]
		\begin{tabular}{|c|c|lll|}
			\hline
			\multirow{2}{*}{Input } & \multirow{2}{*}{Alice's } & \multicolumn{3}{l|}{~~~~~~~~~~~~~~~~~~Bob's decoding}\\ \cline{3-5} 
			string & encoding & \multicolumn{1}{l|}{$1^{st}\to\mathcal{M}_1$ } & \multicolumn{1}{l|}{$2^{nd}\to\mathcal{M}_2$} & $3^{rd}\to\mathcal{M}_3$ \\ \hline\hline
			000 & $\ket{+}\ket{+}$ & \multicolumn{1}{l|}{~~~~~~~~~$0$} & \multicolumn{1}{l|}{~~~~~~~~~$0$} & ~~~~~$\mbox{random}$ \\ \hline
			001 & $\ket{-}\ket{-}$ & \multicolumn{1}{l|}{~~~~~~~~~$0$} & \multicolumn{1}{l|}{~~~~~~~~~$0$} & ~~~~~$\mbox{random}$ \\ \hline
			010 & $\ket{0}\ket{0}$ & \multicolumn{1}{l|}{~~~~~~~~~$0$} & \multicolumn{1}{l|}{~~~~~~~~~$1$} & ~~~~~~~~~~~$0$ \\ \hline
			011 & $\ket{1}\ket{1}$ & \multicolumn{1}{l|}{~~~~~~~~~$0$} & \multicolumn{1}{l|}{~~~~~~~~~$1$} & ~~~~~~~~~~~$1$ \\ \hline 
			100 & $\ket{1}\ket{0}$ & \multicolumn{1}{l|}{~~~~~~~~~$1$} & \multicolumn{1}{l|}{~~~~~~~~~$0$} & ~~~~~~~~~~~$0$ \\ \hline 
			101 & $\ket{0}\ket{1}$ & \multicolumn{1}{l|}{~~~~~~~~~$1$} & \multicolumn{1}{l|}{~~~~~~~~~$0$} & ~~~~~~~~~~~$1$ \\ \hline 
			110 & $\ket{+}\ket{-}$ & \multicolumn{1}{l|}{~~~~~~~~~$1$} & \multicolumn{1}{l|}{~~~~~~~~~$1$} & ~~~~~$\mbox{random}$ \\ \hline 
			111  & $\ket{-}\ket{+}$ & \multicolumn{1}{l|}{~~~~~~~~~$1$} & \multicolumn{1}{l|}{~~~~~~~~~$1$} & ~~~~~$\mbox{random}$ \\ \hline\hline
			\multicolumn{2}{|l|}{$I(a_k:\beta~|~b=k)$}& \multicolumn{1}{l|}{~~~~~~~~~$1$} & \multicolumn{1}{l|}{~~~~~~~~~$1$} & ~~~~~$\approx0.19$\\
			\hline
			\multicolumn{2}{|l|}{~~~~~~~~~~~~$I_3$}&\multicolumn{3}{l|}{$:=\sum_{k=1}^3I(a_k:\beta~|~b=k)\approx2.19>2$}\\
			\hline
		\end{tabular}
		\caption{Violation of IC principle with minimal tensor product composition of two elementary qubits. Since the theory allows only product states it cannot yield any nonlocal correlation. Here $\ket{\pm}:=(\ket{0}\pm\ket{1})/\sqrt{2}$. The bit values in rightmost three columns denote Bob's guess of the corresponding bit, where `random' means he guesses the bit value randomly and be correct in half of the time.}
		\label{tab1}
	\end{table}
\end{center}
At this point it should be noted that the strategy does not invoke any classical correlation which can be shared between Alice and Bob in advance. Therefore the obtained success in Theorem \ref{theo1} might not be the optimal one, and the question of optimality we leave here for future. However, by considering maximal tensor product composition between two qubits our next result provides a perfect strategy for IC-3 game.
\begin{theorem}\label{theo2}
Maximal tensor product composition of two elementary quantum violates IC principle.
\end{theorem}
\begin{proof}
For any $\ket{\psi}_{AB}\in\mathbb{C}^2\otimes\mathbb{C}^2$ we have $\Gamma(\psi_{AB}):=\mathbb{I}_A\otimes\mathrm{T}_B(\psi_{AB})\in\Omega[\mathbb{C}^2\otimes_{\max}\mathbb{C}^2]$, where $\mathbb{I}$ is the identity map, $\mathrm{T}$ is the transpose map, and $\psi:=\ket{\psi}\bra{\psi}$. Alice can use any of these states for encoding her strings. The explicit encodings used by Alice are shown in Table \ref{tab2}. For decoding, Bob is allowed to perform any separable measurement on the composite state received from Alice. By $ZZ$ we denote the separable measurement $\{\ket{00}_{AB}\bra{00}+\ket{11}_{AB}\bra{11},\ket{01}_{AB}\bra{01}+\ket{10}_{AB}\bra{10},\}\equiv\{\pi^{z}_{c}:=\frac{1}{2}\left(\mathbf{1}\otimes\mathbf{1}+\sigma_z\otimes\sigma_z\right),\pi^{z}_{ac}:=\frac{1}{2}\left(\mathbf{1}\otimes\mathbf{1}-\sigma_z\otimes\sigma_z\right)\}$. Similarly we have the separable measurements $XX\equiv\{\pi^{x}_{c}:=\frac{1}{2}\left(\mathbf{1}\otimes\mathbf{1}+\sigma_x\otimes\sigma_x\right),\pi^{x}_{ac}:=\frac{1}{2}\left(\mathbf{1}\otimes\mathbf{1}-\sigma_x\otimes\sigma_x\right)\}$ and $YY\equiv\{\pi^{y}_{c}:=\frac{1}{2}\left(\mathbf{1}\otimes\mathbf{1}+\sigma_y\otimes\sigma_y\right),\pi^{y}_{ac}:=\frac{1}{2}\left(\mathbf{1}\otimes\mathbf{1}-\sigma_y\otimes\sigma_y\right)\}$. For decoding the $1^{st}$, $2^{nd}$, and $3^{rd}$ bit Bob performs $XX$, $YY$, and $ZZ$ measurements respectively, and guesses the bit value as $0~(1)$ if the effect $\pi^\star_c~(\pi^\star_{ac})$ clicks. Expressing Alice's encoding states in Hilbert–Schmidt form: 
\begin{align*}
\Gamma(\phi^+)&=\left(\mathbf{1}\otimes\mathbf{1}+\sigma_x\otimes\sigma_x+\sigma_y\otimes\sigma_y+\sigma_z\otimes\sigma_z\right)/4,\\
\psi^+&=\left(\mathbf{1}\otimes\mathbf{1}+\sigma_x\otimes\sigma_x+\sigma_y\otimes\sigma_y-\sigma_z\otimes\sigma_z\right)/4,\\
\phi^+&=\left(\mathbf{1}\otimes\mathbf{1}+\sigma_x\otimes\sigma_x-\sigma_y\otimes\sigma_y+\sigma_z\otimes\sigma_z\right)/4,\\
\Gamma(\psi^+)&=\left(\mathbf{1}\otimes\mathbf{1}+\sigma_x\otimes\sigma_x-\sigma_y\otimes\sigma_y-\sigma_z\otimes\sigma_z\right)/4,\\
\phi^-&=\left(\mathbf{1}\otimes\mathbf{1}-\sigma_x\otimes\sigma_x+\sigma_y\otimes\sigma_y+\sigma_z\otimes\sigma_z\right)/4,\\
\Gamma(\psi^-)&=\left(\mathbf{1}\otimes\mathbf{1}-\sigma_x\otimes\sigma_x+\sigma_y\otimes\sigma_y-\sigma_z\otimes\sigma_z\right)/4,\\
\Gamma(\phi^-)&=\left(\mathbf{1}\otimes\mathbf{1}-\sigma_x\otimes\sigma_x-\sigma_y\otimes\sigma_y+\sigma_z\otimes\sigma_z\right)/4,\\ 
\psi^-&=\left(\mathbf{1}\otimes\mathbf{1}-\sigma_x\otimes\sigma_x-\sigma_y\otimes\sigma_y-\sigma_z\otimes\sigma_z\right)/4,
\end{align*} 
it becomes evident that this particular encoding-decoding strategy yields a perfect winning condition for the game IC-3. Thus we have $I_3=3$, establishing a violation of IC principle as according to Proposition \ref{prop2} Holevo capacity of the corresponding system is bounded by $2$.
\end{proof}
\begin{center}
	\begin{table}[t!]
		\begin{tabular}{|c|c|lll|}
			\hline
			\multirow{2}{*}{Input } & \multirow{2}{*}{Alice's } & \multicolumn{3}{l|}{~~~~~~~~~~~~~~~~~~Bob's decoding}\\ \cline{3-5} 
			string & encoding & \multicolumn{1}{l|}{$1^{st}\to$~XX} & \multicolumn{1}{l|}{$2^{nd}\to$~YY} & $3^{rd}\to$~ZZ \\ \hline\hline
			000 & $\Gamma(\phi^+)$ & \multicolumn{1}{l|}{~~~~~~~$0$} & \multicolumn{1}{l|}{~~~~~~~$0$} & ~~~~~~~$0$ \\ \hline
			001 & $\psi^+$ & \multicolumn{1}{l|}{~~~~~~~$0$} & \multicolumn{1}{l|}{~~~~~~~$0$} & ~~~~~~~$1$ \\ \hline
			010 & $\phi^+$ & \multicolumn{1}{l|}{~~~~~~~$0$} & \multicolumn{1}{l|}{~~~~~~~$1$} & ~~~~~~~$0$ \\ \hline
			011 & $\Gamma(\psi^+)$ & \multicolumn{1}{l|}{~~~~~~~$0$} & \multicolumn{1}{l|}{~~~~~~~$1$} & ~~~~~~~$1$ \\ \hline 
			100 & $\phi^-$ & \multicolumn{1}{l|}{~~~~~~~$1$} & \multicolumn{1}{l|}{~~~~~~~$0$} & ~~~~~~~$0$ \\ \hline 
			101  & $\Gamma(\psi^-)$ & \multicolumn{1}{l|}{~~~~~~~$1$} & \multicolumn{1}{l|}{~~~~~~~$0$} & ~~~~~~~$1$ \\ \hline
			110 & $\Gamma(\phi^-)$ & \multicolumn{1}{l|}{~~~~~~~$1$} & \multicolumn{1}{l|}{~~~~~~~$1$} & ~~~~~~~$0$ \\ \hline 
			111 & $\psi^-$ & \multicolumn{1}{l|}{~~~~~~~$1$} & \multicolumn{1}{l|}{~~~~~~~$1$} & ~~~~~~~$1$ \\ \hline\hline
			\multicolumn{2}{|l|}{$I(a_k:\beta~|~b=k)$}& \multicolumn{1}{l|}{~~~~~~~$1$} & \multicolumn{1}{l|}{~~~~~~~$1$} & ~~~~~~~$1$\\
			\hline
			\multicolumn{2}{|l|}{~~~~~~~~~~~~$I_3$}&\multicolumn{3}{l|}{$:=\sum_{k=1}^3I(a_k:\beta~|~b=k)=3>2$}\\
			\hline
		\end{tabular}
		\caption{Prefect strategy for IC-$3$ game with communication of two qubits in maximal tensor product. Here, $\ket{\phi^{\pm}}:=(\ket{00}\pm\ket{11})/\sqrt{2}$ and $\ket{\psi^{\pm}}:=(\ket{01}\pm\ket{10})/\sqrt{2}$. Rightmost three columns show that Bob's guess is always correct.  Thus the strategy yields $I(a_k:\beta~|~b=k)=1~\forall~k\in\{1,2,3\}$, and as a result we have $I_3=3$.}
		\label{tab2}
	\end{table}
\end{center}
{\it Discussion.--} While we have excluded the two extreme composition through the violation of IC principle, it would be interesting to exclude all other intermediate compositions except the quantum one in a similar approach. A first step to this direction is to generalize our Proposition \ref{prop1} and \ref{prop2} for such an arbitrary composition which we conjecture to hold true. Then an encoding-decoding strategy in such a composite model yielding a better than quantum success in IC-N game would invalidate that composition to be the physical one. For instance, any composition allowing states and effects used for encoding and decoding either in Theorem \ref{theo1} or in Theorem \ref{theo2} immediately gets excluded as unphysical.

The work of Ref.\cite{Barnum2010} is also worth mentioning, where it has been shown that all the NS correlations possible in any of the bipartite composition of local quantum systems are quantum simulable hence beyond quantum nonlocal correlation is not possible there (see also \cite{Lobo2022}). More precisely, the work in \cite{Barnum2010} does not rule out the existence of POPT states, but shows that they do not generate stronger than quantum correlations in Bell kind of experiment. In fact minimal tensor product composition allows only local correlation in Bell scenario. Despite this, our result show that different such exotic compositions can be rejected as unphysical with the help of IC principles. The power emerges from the fact that IC deals with communication tasks involving both the state space for encodings and effect space for decoding. The beyond quantum exotic nature of state and/or effect space exhibit their implication in the violation of IC principle.

Some recent studies by the present authors \cite{Naik2022,Sen2022} are also worth noting here. By considering the pairwise distinguishability game, where it has been shown that minimal and maximal tensor product compositions allow stronger than quantum correlations in timelike scenario. The feature of `dimension mismatch' \cite{Brunner2014} plays the crucial role there. In that respect, the approach of the present work is logically independent  as it can be argued that `dimension mismatch' is not necessary for violation of IC principle (see Appendix {\bf C}).

{\it Conclusion.--} The notion of composition is one of the guiding tools to fabricate our worldview -- while complex objects are composed of elementary parts, some compositions deem implausible \cite{Penrose1958}. The idea becomes important even while constructing theories in Physics \cite{Hardy2013,Hardy2016,Bhowmik2021}. In this work, we study this particular aspect while considering multiple quantum systems. Interestingly, we show that the principle of Information Causality plays a crucial role in selecting the quantum composition among different mathematical possibilities. In the process Information Causality alone can discard a theory that admits only Bell local correlation. This might make Information Causality champion over the other principles \cite{Navascues2010,Fritz2013,Amaral2014}. Present work thus brings some physical justifying towards quantum composition structure from information principles. The potentiality arises from the communication aspect of those principles since they invoke preparations (for encoding) and measurements (for decoding) of the involved systems and thus becomes more structure sensitive. As for future works, it would be quite interesting to see what other structural aspects of multipartite quantum systems can be rationalized through information principles as the study in the present work is limited to bipartite compositions only.

\begin{acknowledgements}
SGN acknowledges support from the CSIR project 09/0575(15951)/2022-EMR-I. GLS acknowledges support from the CSIR project 09/0575(15830)/2022-EMR-I. MA and MB acknowledge funding from the National Mission in Interdisciplinary Cyber-Physical systems from the Department of Science and Technology through the I-HUB Quantum Technology Foundation (Grant no: I-HUB/PDF/2021-22/008). MB acknowledges support through the research grant of INSPIRE Faculty fellowship from the Department of Science and Technology, Government of India and the start-up research grant from SERB, Department of Science and Technology (Grant no: SRG/2021/000267).
\end{acknowledgements}

\section{Appendix: A}
{\bf Proof of Proposition \ref{prop2}}.
\begin{proof}
Assume that $\mathcal{I}
\left((\mathbb{C}^d)^{\otimes_{\max}k}\right) = N$. Clearly, $N\ge d^k$ since there are $d^k$ product states in $(\mathbb{C}^d)^{\otimes_{\max}k}$ that can be perfectly distinguished by product measurements. Furthermore, since we have assumed $\mathcal{I}
\left((\mathbb{C}^d)^{\otimes_{\max}k}\right) = N$, there exist $N$ POPT states $\{W_1, \cdots , W_N \}$ and a separable measurement $\{E_1,\cdots,E_N~|~\sum_{i=1}^{N} E_i =\mathbf{1}_{d}^{\otimes_{k}}\}$ such that $\tr(E_i W_j)=\delta_{ij} ~~\forall i,j$. Let $Y_i  := \mathbf{1}_{d}^{\otimes_{k}} - W_i$. We will first show that $Y_i$ is an (unnormalized) POPT state, {\it i.e.}, $\bra{a_{1\cdots k}} Y_i \ket{a_{1\cdots k}} \ge 0$ for any choice of product state $\ket{a_{1\cdots k}}:=\ket{a_1 \otimes a_2 \otimes \cdots \otimes a_k}$. Indeed, $\bra{a_{1\cdots k}} Y_i \ket{a_{1\cdots k}}= 1 - \bra{a_{1\cdots k}}  W_i \ket{a_{1\cdots k}}$, and we only need to show that $\bra{a_{1\cdots k}}  W_i \ket{a_{1\cdots k}} \le 1$. But this follows from the fact that $\bra{a_{1\cdots k}}  W_i \ket{a_{1\cdots k}}$ must be a valid probability since $W_i$ is a POPT state and $ \ket{a_1 \otimes a_2 \otimes \cdots \otimes a_k}  \bra{a_1 \otimes a_2 \otimes \cdots \otimes a_k}$ is a valid effect in maximal composition. Now we can proceed in a similar manner as in the proof of Proposition \ref{prop1}, {\it i.e.}
\begin{align*}
d^k &=  \tr(\mathbf{1}_{d}^{\otimes_{k}}) = \sum_{i=1}^{N} \tr(E_i)= \sum_{i=1}^{N} \tr[E_i    \mathbf{1}_{d}^{\otimes_{k}}]\nonumber\\
&= \sum_{i=1}^{N} \tr[E_i   ( W_i + Y_i) ]\nonumber\\
&= \sum_{i=1}^{N} \tr[E_i  W_i ]  + \sum_{i=1}^{N} \tr[E_i Y_i ]\nonumber\\
&\ge \sum_{i=1}^{N} \tr[E_i  W_i] = \sum_{i=1}^{N} \delta_{ii} = N. 
\end{align*}
The inequality follows from the fact that $Y_i$ is POPT, and thus $\tr[E_i Y_i] \ge 0$. Since we already showed that $N\ge d^k$, we conclude that $N= \mathcal{I}\left((\mathbb{C}^d)^{\otimes_{\max}k}\right) = d^k$. Consequently we have $\Theta(\mathbb{C}^d_{\max}):=\lim_{k\to\infty}\frac{1}{k}\log\left[\mathcal{I}((\mathbb{C}^d)^{\otimes_{\max} k})\right]= \lim_{k\to\infty}\frac{1}{k}\log(d^k) = \log d$.
\end{proof}

\section{Appendix: B} 
\begin{figure}[t!]
\centering
\includegraphics[width=0.45\textwidth]{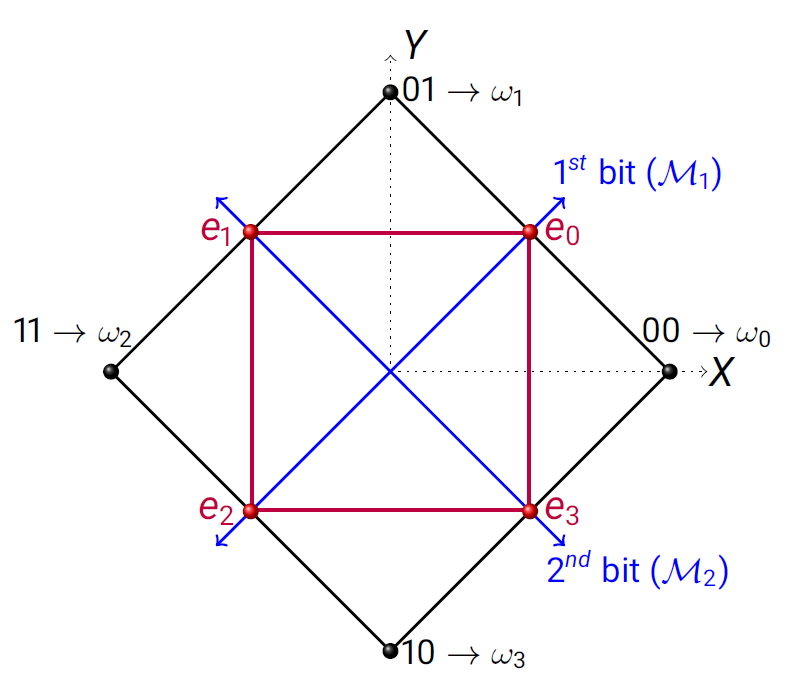}
\caption{Perfect strategy of the IC-$2$ with communication of a square-bit from Alice to Bob. Alice's encodings are shown in the figure. For decoding the first bit Bob performs the measurement $\mathcal{M}_1:=\{e_0,e_2\}$ and guesses the bit value as $0$ if the effect $e_0$ clicks, else he guesses $1$. For the second bit he performs the measurement $\mathcal{M}_2:=\{e_1,e_3\}$ and guesses the bit value as $0$ if the effect $e_3$ clicks, else he guesses $1$. Since Bob always guesses the bit correctly, $I(a_k:\beta|b=k)=1$ for $k\in\{1,2\}$, which in turn results in $I_2=2$, a violation of the IC principle.}
\label{fig3}
\end{figure}
We first briefly recall a toy system having a square shaped state space \cite{Janotta2011}. Its four extremal states are given by $\omega_0=(1,0,1)^\mathrm{T},~\omega_1=(0,1,1)^\mathrm{T},~\omega_2=(-1,0,1)^\mathrm{T},~\&~\omega_3=(0,-1,1)^\mathrm{T}$; and four of its ray extremal effects are $e_0=\frac{1}{2}(1,1,1)^\mathrm{T},~e_1=\frac{1}{2}(-1,1,1)^\mathrm{T},~e_2=\frac{1}{2}(-1,-1,1)^\mathrm{T},~\&~e_3=\frac{1}{2}(1,-1,1)^\mathrm{T}$, while $u=(0,0,1)^\mathrm{T}$ is the unit effect. An interesting feature of this theory is that, unlike quantum theory, it is not self dual, {\it i.e.}, the state cone and effect cone are not identical to each other. Although all the four states of this square-bit theory are pairwise distinguishable, its information capacity as well as Holevo capacity is the same as that of a qubit.

However, it turns out that the square-bit theory violates the IC principle. To see this, consider the IC-$2$ game. With a classical bit or quantum bit communication from Alice to Bob one has $I_2=1$, which satisfies the necessary condition of IC. But with a square-bit communication one can have $I_2=2$, a clear cut violation of the IC principle. The prefect strategy of the IC-$2$ game with a square bit is shown in Fig. \ref{fig3}. For more detail we refer to the work of Ref. \cite{Massar2014}. The violation of IC by square-bit model possibly caused by the lack of `consistency' in mutual information as required in the derivation of the necessary condition of IC. A consistent notion of entropy is not well defined in such a model \cite{Krumm2017,Takakura19}.

\section{Appendix: C}
\begin{figure}[t!]
\centering
\includegraphics[width=0.45\textwidth]{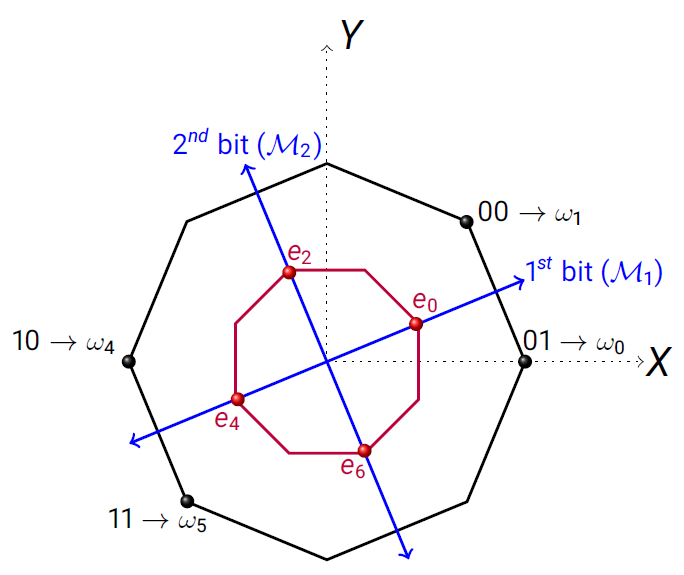}
\caption{Perfect strategy of the IC-$2$ game with communication of an octagon-bit from Alice to Bob. Alice's encodings and Bob's decodings are shown. A straightforward calculation yields $I_2:=\sum_{k=1}^2I(a_k:\beta|b=k)=2-\frac{1}{2\sqrt{2}}+(1-1/\sqrt{2})\log_2(1-1/\sqrt{2})\approx 1+0.1275$, a violation of the IC principle.}
\label{fig4}
\end{figure}
While information capacity
$\mathcal{I}(\mathcal{S})$  of a system $\mathcal{S}$ is the the maximum number of states that can be perfectly distinguished in a single shot measurement, we can also introduce another concept called information dimension of a system \cite{Brunner2014}. Information dimension $\mathcal{I}_{nf}(\mathcal{S})$ of a system $\mathcal{S}$ is the maximum number of states of the system that are perfectly distinguishable pairwise. Arguably it follows that $\Delta\mathcal{I}:=\mathcal{I}_{nf}(\mathcal{S})-\mathcal{I}(\mathcal{S})\ge0$, and whenever $\Delta\mathcal{I}>0$ we say the system exhibits `dimension mismatch'. Importantly, `dimension mismatch' is not necessary for violation of IC principle. To see this, we recall the polygon system with eight extremal states $\omega_i=\left(r_8\cos\frac{2\pi i}{8},r_8\sin\frac{2\pi i}{8},1\right)^{\mathrm{T}}$ and eight extremal effects $e_i=\frac{1}{2}\left(r_8\cos\frac{2\pi(2i+1)}{8},r_8\sin\frac{2\pi (2i+1)}{8},1\right)^{\mathrm{T}}$, with $r_8=\sqrt{\sec\frac{\pi}{n}}~\&~i\in\{0,\cdots,7\}$ \cite{Janotta2011}. Clearly, this system does not exhibit any dimension mismatch. However, as noted in \cite{Massar2014}, it violates the IC principle. For the sake of completeness we depict the explicit protocol in Fig. \ref{fig4}.

\end{document}